\documentclass[review]{elsarticle}
\usepackage[utf8]{inputenc}
\usepackage{amsmath}
\usepackage{amsfonts}
\usepackage{amssymb}
\usepackage{amsthm}
\usepackage{graphicx}
\usepackage{subfigure}
\usepackage[ruled,lined,linesnumbered]{algorithm2e}
\usepackage{xcolor}

\newtheorem{theorem}{Theorem}
\DeclareMathOperator*{\argmin}{arg\,min}
\DeclareMathOperator*{\argmax}{arg\,max}

\journal{Computers \& Operations Research}

\begin{document}

\begin{frontmatter}

\title{Heuristics for $k$-domination models of facility location problems in street networks}

\author{Padraig Corcoran\corref{mycorrespondingauthor}}
\address{School of Computer Science \& Informatics \\
Cardiff University \\
Wales, UK.}

\cortext[mycorrespondingauthor]{Corresponding author}
\ead{corcoranp@cardiff.ac.uk}

\author{Andrei Gagarin}
\address{School of Mathematics \\
Cardiff University \\
Wales, UK.}

\begin{abstract}
\noindent We present new greedy and beam search heuristic methods to find small-size $k$-dominating sets in graphs. The methods are inspired by a new problem formulation which explicitly highlights a certain  structure of the problem. An empirical evaluation of the new methods is done with respect to two existing methods, using instances of graphs corresponding to street networks. The $k$-domination problem with respect to this class of graphs can be used to model real-world facility location problem scenarios. For the classic minimum dominating set ($1$-domination) problem, all except one methods perform similarly, which is due to their equivalence in this particular case. However, for the $k$-domination problem with $k>1$, the new methods outperform the benchmark methods, and the performance gain is more significant for larger values of $k$.
\end{abstract}


\end{frontmatter}

\section{Introduction}
A graph is a mathematical (combinatorial) abstraction that is commonly used to represent many real-world problems. 
A \emph{simple graph} consists of a set of objects called \emph{vertices} and a set of pairwise relations between the objects called \emph{edges}.
For example, a visual scene can be modelled as a graph \cite{xu2017scene}. Similarly, a street network can be modelled as a graph where locations are modelled as vertices and street segments connecting locations are modelled as edges \cite{corcoran2015inferring}. 
Many optimization problems are formulated on graphs. 
These include the shortest path problem, which concerns computing a minimum length path between two vertices in a graph, and the vertex cover problem, which concerns computing a smallest subset of vertices that includes at least one endpoint of every edge. Many real-world problems in turn can be modelled as instances of these graph-theoretic problems. For example, the problem of finding a shortest path in a street network can be modelled as the problem of computing the shortest path in a graph which models that street network.

We consider the minimum $k$-dominating set ($k$-domination) problem in graphs, which is one of the multiple domination problem types (e.g., see \cite{gagarin2013randomized,klasing2004hardness}).
Given a simple graph and a positive integer $k$, the minimum $k$-dominating set ($k$-domination) problem consists in finding a  smallest possible (by cardinality) subset of graph vertices such that each vertex is an element of this subset or is adjacent to at least $k$ elements of this subset. 
Examples and general description for this kind of modelling can be found in Prolegomenon and Chapter $1$ of the classic book on domination in graphs \cite{Fundamentals1998}.
Thai et al. \cite{thai2007approximation} modelled the problem of computing a virtual backbone in a wireless ad-hoc or sensor network as a $k$-domination problem. Gagarin et al. \cite{gagarin2018multiple} modelled the problem of optimizing the placement of electrical vehicle charging stations as a $k$-domination problem. Also, Khomami et al. \cite{khomami2018minimum} modelled the problem of maximizing influence in a social network as a $k$-domination problem.

The $k$-domination problem has been proven to be $\mathcal{NP}$-hard \cite{Jacob1989}, even, e.g., in split graphs \cite{Chang2013}. As a consequence, unless the problem instance is reasonably small, one generally cannot use an exact method to compute an optimal solution in reasonable time (e.g., see the state-of-the-art deterministic algorithms and computational results in \cite{Wendy1,Wendy2}). Therefore, heuristic methods are normally used 
to find small-size $k$-dominating sets in reasonable time, assuming the solution may be suboptimal. 

In this article we propose a novel formulation of the $k$-domination problem. This new formulation makes explicit important structure in the problem, which is not present in existing formulations. We subsequently propose two heuristic methods for solving this problem, which exploit this structure. The methods in question use greedy and beam search approach ideas. We empirically evaluate these two methods with respect to street network reachability graphs. The $k$-domination problem with respect to this class of graphs can be used to model facility location problems in street networks \cite{gagarin2018multiple}.

The remainder of this paper is structured as follows. In Section \ref{sec:background} we review existing solutions to the $k$-domination problem. In Section \ref{sec:method} we formally define the $k$-domination problem and the proposed novel problem formulation. In this section we also describe the proposed heuristic methods for solving this problem. In Section \ref{sec:results} we present an experimental evaluation of the proposed methods with respect to existing baseline methods on street network reachability graphs. Finally, in Section \ref{sec:conclusions} we draw some conclusions from this work and discuss some possible directions for future research.

\section{Related Works}
\label{sec:background}
In this section we review existing methods for computing solutions to the $k$-domination problem. We focus exclusively on the case where the graphs in question are unweighted and undirected. The methods described in this section do not naturally generalise to other types of graphs and instead specialized methods must be considered, e.g., see \cite{wang2018fast}. A number of authors have proposed methods for computing solutions to different variants of the $k$-domination problem. For example, Klasing and Laforest \cite{klasing2004hardness} considered the $k$-tuple domination problem which is a more constrained variation of the $k$-domination problem. Shang et al. \cite{shang2007algorithms} proposed a method for computing a $k$-tuple dominating set which is also $m$-connected.

The $k$-domination problem is a classic optimization problem. As a consequence, a large number of methods for solving this problem have been proposed. These methods can broadly be distinguished with respect to the following five features. The first feature concerns whether the method in question is designed for the classic minimum dominating set problem, i.e. $1$-domination problem in our more general context ($k=1$). The second feature concerns whether the method in question automatically generalizes to the cases where $k>1$. The final three features concern whether the method in question uses a greedy search heuristic, a metaheuristic or an exact method to determine a solution. Both greedy search heuristic and metaheuristic methods attempt to compute a useful solution in a reasonable amount of time, where this solution may be not optimal. On the other hand, exact methods attempt to compute an optimal solution. Table \ref{table:existing_methods} presents a summary of existing methods for the $k$-domination problem with respect to these five features.

The distinction with respect to whether a method is applicable only to the case $k=1$ or automatically generalizes to the case  $k>1$ is particularly important in the context of this work. Therefore, in Sections \ref{sec:background_k_1} and \ref{sec:background_k_n} we respectively review methods belonging to these two categories.

\begin{table}
\centering
\begin{tabular}{ c|c|c|c|c|c } 
 \hline
  & $k=1$ & $k \geq 1$ & Greedy & Metaheuristic & Exact \\
  & & & Search & & Method \\
 \hline
 Parekh \cite{parekh1991analysis} & \checkmark & & \checkmark & & \\
 Sanchis \cite{sanchis2002experimental} & \checkmark & & \checkmark & & \\
 Eubank et al. \cite{eubank2004structural} & \checkmark & & \checkmark & & \\
 Chellali et al. \cite{chellali2012k} & \checkmark & \checkmark & \checkmark & & \\
 Hedar et al. \cite{hedar2010hybrid} & \checkmark & & & \checkmark & \\
 Hedar et al. \cite{hedar2012simulated} & \checkmark & & & \checkmark & \\
 Ho et al. \cite{ho2006enhanced} & \checkmark & & & \checkmark & \\
 Nehez et al. \cite{nehez2015comparison} & \checkmark & & & & \checkmark \\
 Bird \cite{Wendy1} & \checkmark & & & & \checkmark \\
 Assadian \cite{Wendy2} & \checkmark & & & & \checkmark \\
 Couture et al. \cite{couture2008incremental} & & \checkmark & \checkmark & & \\
 Gagarin et al. \cite{gagarin2013randomized} & & \checkmark & \checkmark & & \\
 Gagarin et al. \cite{gagarin2018multiple} & & \checkmark & \checkmark & & \checkmark \\
 \hline
\end{tabular}
\caption{Methods for computing solutions to the $k$-domination problem are distinguished with respect to five features.}
\label{table:existing_methods}
\end{table}

\subsection{Searching for minimum dominating sets ($k=1$)}
\label{sec:background_k_1}
Existing solution methods for the classic minimum dominating set problem, i.e. the $k$-domination problem where $k=1$, can be broadly divided into greedy search heuristic, metaheuristic, and exact solution (deterministic) methods. We now review methods belonging to each of these categories in turn.

\subsubsection*{Greedy Search Heuristic Methods}
The first and most commonly used standard greedy search heuristic for computing dominating sets ($k=1$) is described in Parekh \cite{parekh1991analysis}. The method initializes a set $D$ to be the empty set and iteratively adds vertices to $D$ until it forms a dominating set. The vertex added to $D$ at each iteration is determined by selecting a vertex from the set of vertices whose neighbourhood contains a maximum number of vertices currently not dominated. In this context, a vertex is not dominated if it is not an element of the set $D$ and not adjacent to any vertex in $D$.

Sanchis \cite{sanchis2002experimental} evaluated four greedy search heuristic methods for computing small-size dominating sets. The first method is entitled \textit{Greedy}. This method initializes a set $D$ to be the empty set and iteratively adds vertices to $D$ until it forms a dominating set. The vertex added to $D$ at each step is determined by selecting uniformly at random a vertex from the set of vertices whose neighbourhood contains a maximum number of vertices currently not dominated. This method is similar to that described by Parekh \cite{parekh1991analysis} but with the addition of randomization in vertex selection. The second method is entitled \textit{Greedy\_Rev}. This method initializes a set $D$ to equal the set of graph vertices in question and iteratively removes vertices from $D$ until no further vertex can be removed while still maintaining the property that $D$ is a dominating set. The vertex removed from $D$ at each step is determined by selecting uniformly at random a vertex from the set of vertices which are eligible to be removed and have the maximum degree. The third method is entitled \textit{Greedy\_Ran}. This method is similar to that entitled \textit{Greedy} with the exception that the vertex added to $D$ at each step is determined by selecting a vertex with probability proportional to the number of adjacent vertices currently not dominated. The final method is entitled \textit{Greedy\_Vote}. This method initializes a set $D$ to be the empty set and iteratively adds vertices to $D$ until it forms a dominating set. The vertex added to $D$ at each step is determined by selecting a vertex with probability proportional to the number of neighbours of its neighbours currently not dominated. The author evaluated the four above methods on two different classes of graphs and found the \textit{Greedy} and \textit{Greedy\_Vote} methods to perform best.

Eubank et al. \cite{eubank2004structural} evaluated five greedy search heuristic methods for finding dominating sets. The first method is called \textit{RegularGreedy} and is the same as the standard \cite{parekh1991analysis}. The second method is named \textit{FastGreedy}. This method initializes a set $D$ to the empty set. It then iterates over the vertices in the graph considering vertices of greater degree first and adding each vertex to the set $D$ until it forms a dominating set. The third method is entitled \textit{VRegularGreedy}. This method initializes a set $D$ to be the set of all neighbours of vertices of degree $1$. It subsequently applies the standard greedy approach \cite{parekh1991analysis} to the graph induced by vertices currently not in $D$. The fourth and fifth methods are called \textit{FastGreedy-1} and \textit{FastGreedy-2}. Both methods are slight variations of the \textit{FastGreedy} method described above. The authors evaluated the above five methods on a number of real-world social networks and random graphs. They found that the methods \textit{RegularGreedy} and \textit{VRegularGreedy} performed best.

Chellali et al. \cite{chellali2012k} proposed a greedy search heuristic method which initializes a set $D$ to be the empty set and iteratively adds vertices to $D$ until it forms a dominating set. The vertex added to $D$ at each step is determined by selecting uniformly at random a vertex from the set of vertices currently not dominated. This method is implemented in the NetworkX software library which is a highly popular Python software library for graph analysis \cite{hagberg2008exploring}.

It is important to note that many of the greedy search heuristic methods reviewed above are also randomized methods. This combined with the general low computational complexity of these methods means that they can be applied to a given problem instance a large number of times with the best solution obtained being returned.

\subsubsection*{Metaheuristic Methods}
Hedar and Ismail \cite{hedar2010hybrid} proposed a number of genetic algorithms for finding dominating sets and evaluated them on a set of random graphs. The same authors later proposed a simulated annealing method to search for dominating sets \cite{hedar2012simulated}. They experimentally tested this method with respect to a stochastic local search method, the genetic algorithm of  \cite{hedar2010hybrid}, and the method entitled \textit{Greedy} proposed by Sanchis \cite{sanchis2002experimental} on a set of random graphs. Their experiments show the simulated annealing and the genetic algorithm of \cite{hedar2010hybrid} to perform best.

Ho et al. \cite{ho2006enhanced} proposed a number of ant colony optimization methods for computing dominating sets. The authors evaluated these methods and a genetic algorithm on a set of random graphs. They found that an ant colony optimization method outperforms the genetic algorithm.

\subsubsection*{Exact Methods}
Nehez et al. \cite{nehez2015comparison} proposed an integer linear programming (ILP) method for computing dominating sets. The authors evaluated this method against a randomized local search method and the standard greedy search heuristic \cite{parekh1991analysis} on a number of real-world graphs. The authors found that the ILP approach performed best but did not scale to large graphs. The same was shown by computational experiments in Gagarin and Corcoran \cite{gagarin2018multiple}, where an ILP formulation is described for a more general $k$-domination problem scenario.

The state-of-the-art deterministic search methods for dominating sets in graphs have been recently developed and described by Bird \cite{Wendy1} and Assadian \cite{Wendy2}. The methods are based on backtracking, and the experimental results indicate that they are not likely to be practical for graphs with more than several hundred vertices. 

\subsection{Searching for small $k$-dominating sets ($k \geq 1$)}
\label{sec:background_k_n}
The more general $k$-domination problem, where $k \geq 1$, is less well studied than the classic minimum dominating set problem ($k=1$). In fact, only a few solution methods described in the previous section generalize to solve the $k$-domination problem for any $k\ge 1$. 
These solution methods can be broadly divided into greedy search heuristics and exact (deterministic) algorithms.

\subsubsection*{Greedy Search Heuristic Methods}
A generalization of the standard greedy algorithm (\cite{parekh1991analysis}) for computing $k$-dominating sets is described in \cite{gagarin2018multiple}. Specifically, this method initializes a set $D$ to be the empty set and iteratively adds vertices to $D$ until it forms a $k$-dominating set. The vertex added to $D$ at each step is determined by selecting uniformly at random a vertex from the set of vertices whose neighbourhood contains a maximum number of vertices currently not dominated enough.
In a certain sense, this simple greedy algorithm is inspired by the greedy approach to find $k$-tuple dominating sets in Klasing et al. \cite{klasing2004hardness}.

Couture et al. \cite{couture2008incremental} proposed a method which first computes a dominating set ($k=1$) by finding a maximal independent set in a graph. Next, their algorithm computes a maximal independent set for the vertices that are currently not $2$-dominated and adds those vertices to the dominating set to form a $2$-dominating set. This procedure is repeated $k$ times until a $k$-dominating set is found.

Gagarin et al. \cite{gagarin2013randomized} proposed a randomized algorithm which initializes a set $D$ as a random subset of graph vertices and then iteratively adds other vertices to $D$ if they are not dominated enough, mentioning some greedy ideas. The probability to initialize set $D$ randomly is shown to be optimal in general graphs. However, this probability had to be adjusted experimentally for graphs corresponding to real-world road networks in \cite{gagarin2018multiple}. The last paper also experimentally compares the randomized approach to the basic greedy heuristic.

\subsubsection*{Exact (deterministic) methods}
An ILP problem formulation for computing $k$-dominating sets is described in \cite{gagarin2018multiple}. 
The experiments show that this method clearly does not scale to the size of two main graphs considered in the paper.
Therefore a greedy search heuristic remains one of the main optimization tools in that research.
Also, the computational results in \cite{gagarin2018multiple} show that the ILP formulation solution approach scales less well for larger values of $k$.

\section{New Heuristic Search Methods for $k$-Domination}
\label{sec:method}
In this section we formally define the $k$-domination problem and present a novel formulation of this problem. This formulation is in turn used to develop two novel methods for solving the problem which use greedy and beam search heuristic ideas. 

We consider simple graphs $G=(V,E)$, where $V$ is a set of vertices and $E$ is a set of edges. 
Given a vertex $v\in V$, the \textit{open neighbourhood} of $v$ is the set of all its neighbours in $G$, i.e. all vertices adjacent to $v$. It is denoted by $N(v)$. The \textit{closed neighbourhood} of $v$ is $N(v)\cup \{v\}$, it is denoted by $N[v]$.
For a given positive integer $k$, a \emph{$k$-dominating set} of $G$ is a set $D \subseteq V$ such that each $v \in V$ is either an element of $D$ or is adjacent to at least $k$ elements of $D$. The \emph{$k$-domination problem} concerns finding a $k$-dominating set of $G$ which is as small as possible. This is formally defined as the following optimization problem:

\begin{equation}
\label{eq:dom_opt}
\begin{aligned}
& \argmin_{D\subseteq V} & & \vert D \vert \\
& \text{subject to} & & \forall \ v \in V \setminus D, \ \vert N(v) \cap D \vert \geq k
\end{aligned}
\end{equation}

In general, solving this optimization problem is $\mathcal{NP}$-hard \cite{Jacob1989,Chang2013}. For any $D\subseteq V$ and $v\in V$, we define the parameter $C(D,v)$, which indicates a level of coverage of neighbours of $v$ with respect to the set $D$ in $G$ as follows: 

\begin{equation}
\label{eq:dom_opt_C}
C(D, v) = \text{min} (k, \vert N(v) \cap D \vert) 
\end{equation}

\noindent Then the optimization problem (\ref{eq:dom_opt}) can be formulated as the optimization problem (\ref{eq:dom_opt_re}). We prove this is an equivalent problem formulation in Theorem \ref{thm:mul_re}.

\begin{equation}
\label{eq:dom_opt_re}
\begin{aligned}
& \argmax_{D\subseteq V} & & \sum_{v \in V \setminus D} C(D, v) \\
& \text{subject to} & & \forall \ v \in V \setminus D, \ \vert N(v) \cap D \vert \geq k
\end{aligned}
\end{equation}

\begin{theorem}
\label{thm:mul_re}
Given a graph $G$, a solution $D$ to the optimization problem defined in (\ref{eq:dom_opt_re}) is a minimum size $k$-dominating set in $G$.
\end{theorem}
 
\begin{proof}
A set $D\subseteq V$ satisfying the constraints in (\ref{eq:dom_opt_re}) is a $k$-dominating set: each vertex $v \in V \setminus D$ is adjacent to at least $k$ elements in $D$. Therefore, $C(D,v)=\text{min} (k, \vert N(v) \cap D \vert) = k$ for all $v \in V \setminus D$. In turn, the value of the objective function in (\ref{eq:dom_opt_re}) equals $k (\vert V \vert - \vert D \vert)$. Since $\vert V \vert$ is a constant, the objective function is maximized when $D$ is a minimum size $k$-dominating set.
In other words, the optimization problem (\ref{eq:dom_opt_re}) is equivalent to maximizing $|V\backslash D|$ in $G$.
\end{proof}

Since the $k$-domination problem is $\mathcal{NP}$-hard, one normally must use heuristic methods 
to find a reasonably small $k$-dominating set (\cite{gagarin2013randomized}), sacrificing quality of the solution to a reasonable computational time (\cite{gagarin2018multiple,Wendy1,Wendy2}). The problem formulation in (\ref{eq:dom_opt_re}) uses the parameter $C(D,v)$ from (\ref{eq:dom_opt_C}) to model the level to which vertices are dominated by a set $D$. This contrasts with the formulation (\ref{eq:dom_opt}), which simply 
requires satisfaction of the constraints.
The additional information incorporated in formulation (\ref{eq:dom_opt_re}) can potentially be exploited by heuristic methods to make locally optimal decisions better.

In the following subsections we describe two heuristic solution methods for the $k$-domination problem which use formulation (\ref{eq:dom_opt_re}). These methods are based on greedy and beam search heuristics ideas.

\subsection{New Greedy Search Heuristic}
\label{sec:greedy}
The problem formulation in (\ref{eq:dom_opt_re}) motivates Algorithm \ref{alg:rand_gred_dom} which is a greedy search heuristic. The algorithm takes as input a graph $G=(V,E)$ and a positive integer $k$, and computes a $k$-dominating set $D$ for $G$. The algorithm initializes $D$ to be the empty set (line \ref{alg:rand_gred_dom:initial}). Next, it iteratively adds vertices to $D$ until it forms a $k$-dominating set. The vertex added at each step is determined by selecting uniformly at random a vertex from the set of vertices whose addition maximizes the unconstrained objective function in (\ref{eq:dom_opt_re}) (see lines \ref{alg:rand_gred_dom:start} to \ref{alg:rand_gred_dom:end}).

Convergence of Algorithm \ref{alg:rand_gred_dom} to a $k$-dominating set is a consequence of the fact that, unless a $k$-dominating set is formed earlier, the algorithm will converge to the case where $D=V$, which is trivially $k$-dominating in $G$. Algorithm \ref{alg:rand_gred_dom} is also a randomized algorithm: in each iteration, if several vertices can increase the objective function value by the same maximum amount, one of these vertices is selected uniformly at random for addition to $D$.

\begin{algorithm}
\label{alg:rand_gred_dom}
\caption{Greedy Search Heuristic}
 \KwIn{A graph $G=(V,E)$, a positive integer $k$.}
 \KwOut{A $k$-dominating set $D$ of $G$.}
 \BlankLine
\Begin{
Initialize $D= \lbrace \rbrace$ \\ \label{alg:rand_gred_dom:initial} 
 \While{$ |\lbrace v\in V \backslash D : |N(v)\cap D|<k \rbrace| > 0 $}{ \label{alg:rand_gred_dom:start}
		Find\,\ $\displaystyle U = \argmax_{u \in V \setminus D} \sum_{v \in V \setminus (D \cup \{u\})} C(D \cup \{u\}, v)$\\
		Sample $u \in U$ using a uniform distribution \\
		Put\,\ $D = D \cup \{u\}$
	} \label{alg:rand_gred_dom:end}
 return $D$
}
\end{algorithm}

To better explain effectiveness and efficiency of the heuristic ideas used in Algorithm \ref{alg:rand_gred_dom}, we define the difference function $\Delta(D,u)$ (\ref{eq:delta_C}) which represents the change in the objective function $C$ in (\ref{eq:dom_opt_re}) when a new vertex vertex $u$ is added to a given set $D\subset V$ ($u\not\in D$):
\begin{equation}
\label{eq:delta_C}
\begin{split}
\Delta(D, u) &= \sum_{v \in V \setminus (D \cup \{u\})} C(D \cup \{u\}, v) - \sum_{v \in V \setminus D} C(D, v) \\
 & = | \{ v \in N(u) : v \notin D, | N(v) \cap D  | < k \} | - \text{min}(k, | N(u) \cap D |)
\end{split}
\end{equation}
Now we have
\begin{equation}
\label{eq:delta_eq}
\begin{split}
& \max_{u \in V \setminus D}\ \ \Delta(D, u)\\
 & = \max_{u \in V \setminus D} \left(\sum_{v \in V \setminus (D \cup \{u\})} C(D \cup \{u\}, v)\right) - \sum_{v \in V \setminus D} C(D, v),\\ 
 \end{split}
\end{equation}
where the sum $\displaystyle \sum_{v \in V \setminus D} C(D, v)$ is constant for a given set $D\subset V$.
Therefore, in Algorithm \ref{alg:rand_gred_dom}, we have 
\begin{equation}
\label{eq:delta_eq_1}
\begin{split}
U & = \argmax_{u \in V \setminus D} \sum_{v \in V \setminus (D \cup \{u\})} C(D \cup \{u\}, v) \\
  & =  \argmax_{u \in V \setminus D}\ \ \Delta(D, u)
\end{split}
\end{equation}
which is used when implementing Algorithm \ref{alg:rand_gred_dom}. The computational complexity of Algorithm \ref{alg:rand_gred_dom} can be analyzed as follows.

\begin{theorem}
\label{thm:greedy_complexity}
Algorithm \ref{alg:rand_gred_dom} finds a $k$-dominating set in $G$ in $\mathcal{O}(n^3)$ time, $n=|V|$.
\end{theorem}
 
\begin{proof}
In the worst case, the while loop will terminate after $n$ iterations when all vertices have been added to $D$. Each iteration of the while loop examines all vertices currently not in the set $D$. For each such vertex $u\in V\backslash D$, the change in the objective function following its addition is computed using the difference function $\Delta(D,u)$ from (\ref{eq:delta_C}). Evaluating the difference function (\ref{eq:delta_C}) in the worst case when $|N(u)|\in \Theta(n)$ can be done in $O(n)$ time. Therefore each iteration of the while loop takes $O(n^2)$ time.
\end{proof}

To develop a better intuition of how Algorithm \ref{alg:rand_gred_dom} works, notice that at each step the algorithm does not determine the vertex added to $D$ solely based on the number of vertices currently not dominated in its open or closed neighbourhood. In this context, a vertex is not dominated if it is not an element of $D$ and not adjacent to at least $k$ vertices in $D$. The algorithm also considers how much  the vertices in question are already dominated. Specifically, all other things being equal, a vertex which is currently least dominated will be added to $D$ because it contributes least to the sum $\sum_{v \in V \setminus D} C(D, v)$ (note that the sum is over vertices currently not in the set $D$, see definition (\ref{eq:delta_C})). For example, if $C(D,v) < k$ for all $v\in V\backslash D$, we have $\sum_{v \in V \setminus (D \cup \{u\})} C(D \cup \{u\}, v) = \sum_{v \in V \setminus D} C(D, v) + |N(u)\cap (V \setminus D)| - |N(u)\cap D|$ for each $u\in V\backslash D$.

To illustrate this concept, consider the graph displayed in Figure \ref{fig:three_vertex_graph}. Suppose we wish to compute a $2$-dominating set and are given $D=\lbrace a \rbrace$. If the next vertex added to $D$ is determined solely based on the number of vertices currently not dominated in its corresponding open neighbourhood, vertices $b$ and $c$ are equally likely to be added -- each of the neighbourhoods contains exactly one vertex currently not dominated.
This may give a suboptimal result because adding vertex $b$ will not result in a $2$-dominating set. The result is the same if we consider the closed neighbourhood instead of the open neighbourhood of the vertices. On the other hand, the proposed algorithm will add to $D$ the vertex $c$, which provides the optimal solution.

\begin{figure}
\begin{center}
\vspace{5mm}
\includegraphics[width=55mm]{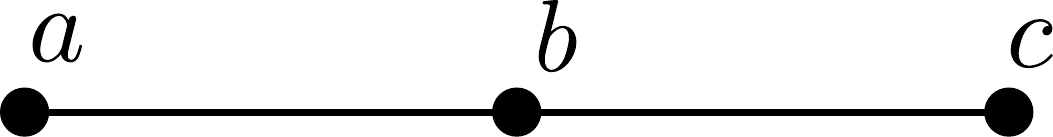}
\caption{An illustrative graph.}
\label{fig:three_vertex_graph}
\end{center}
\end{figure} 

\subsection{Beam Search Heuristic}
\label{sec:beam}
Algorithm \ref{alg:rand_gred_dom} is a greedy algorithm, and consequently it may converge to a suboptimal solution. To illustrate this, consider again the graph displayed in Figure \ref{fig:three_vertex_graph} and the problem of computing a $2$-dominating set. The optimal $2$-dominating set for this problem instance is $\lbrace a, c \rbrace$. However, applying Algorithm \ref{alg:rand_gred_dom} to this problem will first add vertex $b$ to $D$, followed by vertices $a$ and $c$, to give a larger $2$-dominating set $\lbrace a, b, c \rbrace$.

To overcome this limitation, we propose a generalization of Algorithm \ref{alg:rand_gred_dom} which uses a beam search heuristic instead of a pure greedy approach. A beam search algorithm is an iterative search method which at each step maintains a set of best intermediate or partial solutions \cite{rusnor2009}. The maximum size of this set is a constant hyper-parameter called the \emph{beam width}. The proposed beam search with the beam width of one is equivalent to the greedy search of Section \ref{sec:greedy}. The general beam search heuristic implemented in this work is described in Algorithm \ref{alg:rand_beam_dom}. This algorithm takes as input a graph $G=(V,E)$ and two positive integer parameters $k$ and $b$, and, using a beam of width $b$, finds a $k$-dominating set $D$ for $G$ by considering a collection $S$ of vertex subsets of $G$.

The algorithm initializes $D$ to be the empty set (line \ref{alg:rand_beam_dom:initial_D}) and $S$ to be a list containing a single partial solution corresponding to the empty set. In this context, a partial or intermediate solution is a subset of graph vertices. Next, the algorithm iteratively expands all subsets of vertices in $S$ (lines \ref{alg:rand_beam_dom:exp_1} to \ref{alg:rand_beam_dom:exp_2}). Each partial solution $s\in S$, $s\subseteq V$, is expanded to form a set of partial solutions by adding one vertex currently not in $s$ in all possible ways, i.e. like in an exhaustive search. 
For example, consider the graph in Figure \ref{fig:three_vertex_graph} and a partial solution set $\lbrace b \rbrace$. 
Expanding this vertex subset in all possible ways gives two partial solutions $\lbrace a, b \rbrace, \lbrace b, c \rbrace$. Similarly, expanding the empty set $\lbrace \rbrace$ in the context of the same graph gives three partial solutions $\lbrace a \rbrace, \lbrace b \rbrace, \lbrace c \rbrace$. The algorithm next removes copies of partial solutions from the list of subsets $S$ (line \ref{alg:rand_beam_dom:duplic}). The list of intermediate solutions is then sorted in non-ascending (descending) order with respect to the objective function value in (\ref{eq:dom_opt_re}). If several partial solutions have the same objective function value, they are ordered randomly in $S$ (line \ref{alg:rand_beam_dom:sort}). Then the top $b$ partial solutions are retained in $S$ (line \ref{alg:rand_beam_dom:retain}). 
Finally, if there is a $k$-dominating set $s$ in $S$, i.e. a set $s$ satisfying constraints in (\ref{eq:dom_opt_re}), we put $D=s$, and it is returned as a solution to the problem (lines \ref{alg:rand_beam_dom:check_s} and \ref{alg:rand_beam_dom:return} respectively).

Convergence of Algorithm \ref{alg:rand_beam_dom} to a $k$-dominating set is a consequence of increasing cardinality of partial solutions in iteration and the fact that, unless a $k$-dominating set is found earlier, the algorithm will converge to the case where $D=V$, which is $k$-dominating. Algorithm \ref{alg:rand_beam_dom} is also a randomized algorithm: as stated above, when sorting is performed, if several partial solutions have the same unconstrained objective function value, they are ordered randomly. The computational complexity of Algorithm \ref{alg:rand_beam_dom} is stated in Theorem \ref{thm:beam_complexity}.

\begin{algorithm}
\label{alg:rand_beam_dom}
\caption{Beam Search Heuristic}
 \KwIn{A graph $G=(V,E)$, positive integers $k$ and $b$.}
 \KwOut{A $k$-dominating set $D$ of $G$.}
 \BlankLine
\Begin{
Initialize $D = \lbrace \rbrace$ \\ \label{alg:rand_beam_dom:initial_D}
Initialize $S = [\lbrace \rbrace]$ \\ \label{alg:rand_beam_dom:initial_S}
 \While{$ D = \lbrace \rbrace $}{ \label{alg:rand_beam_dom:start}
 		$S' = []$ \\ \label{alg:rand_beam_dom:exp_1}
 		\For{$s \in S$}{
 			$S' = S' \cup \text{expand}(s)$
 		}
		$S = S'$ \\ \label{alg:rand_beam_dom:exp_2}
		remove\_duplicates($S$) \\ \label{alg:rand_beam_dom:duplic}
		sort\_descending($S$) \\ \label{alg:rand_beam_dom:sort}
		$S=S[1 \dots b]$ \\ \label{alg:rand_beam_dom:retain}
		\For{$s \in S$}{
			\If{$s$ is $k$-dominating}{
				$D=s$ \\ \label{alg:rand_beam_dom:check_s}
			}
		}
	}
 return $D$ \label{alg:rand_beam_dom:return}
}
\end{algorithm}

\begin{theorem}
\label{thm:beam_complexity}
Algorithm \ref{alg:rand_beam_dom} finds a $k$-dominating set in $\mathcal{O}(b^2n^3)$ time, where $n=|V|$ and $b$ is the beam width.
\end{theorem}
 
\begin{proof}
In the worst case, the while loop will terminate after $n$ iterations. Each iteration of the while loop performs an expansion of intermediate solutions in $S$, which contains $O(b)$ subsets of vertices. 
Expanding each individual partial solution takes $O(n)$ steps. Then the list $S'$ will contain $O(nb)$ elements. 
Checking for copies of subsets in line \ref{alg:rand_beam_dom:duplic} can be done in $O(b^2n^2)$ time.
Evaluating a single element in the resulting list with respect to the objective function can be done in $O(n)$ time. 
Therefore, evaluating and sorting all elements in this list with respect to the objective function takes $O(bn^2 + bn \log bn)$ steps. The overall computational time complexity is therefore $O(nbn + nb^2n^2 + nbn^2 + nbn\log bn) = O(b^2n^3)$.
\end{proof}

Selecting the beam width hyper-parameter for Algorithm \ref{alg:rand_beam_dom} represents a trade-off between computational complexity and the quality of obtained solution. That is, a larger beam width results in a better exploration of the partial solution space and the potential for finding a better solution. However, a larger beam width also results in higher computational complexity. To illustrate this, consider again the graph in Figure \ref{fig:three_vertex_graph} and the problem of computing a $2$-dominating set. Applying Algorithm \ref{alg:rand_beam_dom} to this graph with the beam width of one returns the $2$-dominating set $\lbrace a, b, c \rbrace$. On the other hand, applying this algorithm to the same graph with the beam width of three returns the smaller $2$-dominating set $\lbrace a, c \rbrace$. As mentioned earlier, Algorithm \ref{alg:rand_beam_dom} with the beam width equal to one reduces to Algorithm \ref{alg:rand_gred_dom}.

In Theorem \ref{thm:mul_equ_methods}, we establish a relationship between Algorithm \ref{alg:rand_beam_dom} and a standard greedy approach for computing dominating sets ($k=1$) described in \cite{parekh1991analysis,gagarin2018multiple}. 
Recall that the standard greedy algorithm initializes a set $D$ to be the empty set and iteratively adds vertices to $D$ until it forms a dominating set. The vertex added to $D$ at each step is determined by selecting uniformly at random a vertex from the set of vertices whose closed neighbourhood contains a maximum number of vertices currently not dominated.

\begin{theorem}
\label{thm:mul_equ_methods}
When computing a dominating set ($k=1$), Algorithm \ref{alg:rand_beam_dom} with the beam width of one is equivalent to the standard greedy algorithm \cite{parekh1991analysis,gagarin2018multiple}.
\end{theorem}

\begin{proof}
For the beam width of one, Algorithm \ref{alg:rand_beam_dom} behaves greedily and at each step selects the vertex which maximizes the change in the unconstrained objective function of (\ref{eq:dom_opt_re}). It is possible to see that for the case $k=1$, the change in the objective function value by adding a vertex $v$ equals the number of not dominated vertices in the closed neighbourhood of $v$ minus one. In other words, the selection criteria functions used by Algorithm \ref{alg:rand_beam_dom} and the standard greedy algorithm to rank vertices only differ by a constant value of minus one, implying the vertex rankings are the same. 

Consider two possible mutually exclusive cases corresponding to $v$ being currently dominated or not.
If $v$ is currently not dominated, the change in the objective function value by adding $v$ equals the number of not dominated vertices in the open neighbourhood of $v$, which is the number of not dominated vertices in the closed neighbourhood of $v$ minus one.
If $v$ is currently dominated, the change in the objective function value by adding $v$ equals the number of not dominated vertices in the open neighbourhood of $v$ minus one, which is, in this case, the number of not dominated vertices in the closed neighbourhood of $v$ minus one.
\end{proof}

\section{Computational Results and Analysis}
\label{sec:results}
In this section we present an empirical evaluation of the two proposed heuristic methods for computing $k$-dominating sets with respect to two baseline methods. We perform this evaluation using a set of graphs corresponding to street network reachability graphs. The $k$-domination problem with respect to this class of graphs can be used to model facility location problems in street networks \cite{gagarin2018multiple}.

The remainder of this section is structured as follows. In Section \ref{sec:results:reach} we formally define the concept of a street network reachability graph and the corresponding facility location problem. In Section \ref{sec:results:stn} we present details of the street networks used in this evaluation. Section \ref{sec:results:baseline} describes the baseline methods against which the proposed methods are evaluated. Finally, in Section \ref{sec:results:comparsion} we present empirical results of our evaluation.

\subsection{Reachability Graphs and Facility Location}
\label{sec:results:reach}
A street network can be modelled as a weighted undirected graph $G^s=(V^s,E^s, w:E^s\rightarrow \mathbb{R})$, where the set of vertices $V^s$ corresponds to road intersections and dead-ends, while the set of edges $E^s$ corresponds to road segments connecting these vertices. The weight function $w$ assigns to each edge the length of the corresponding road segment measured in meters \cite{corcoran2013characterising}. The street network of Cardiff city modelled as such a graph is displayed in Figure \ref{fig:cardiff_st}.

Given a street network graph $G^s=(V^s,E^s, w:E^s\rightarrow \mathbb{R})$, we define its reachability graph $G^r_t=(V^r,E^r_t)$ as a simple unweighted graph with $V^r = V^s$ and $(u, v)\in E^r_t$ if and only if the length of shortest path (distance) between the corresponding vertices $u$ and $v$ in $G^s$ is less than a specified reachability threshold of $t$ meters \cite{gagarin2018multiple}. The reachability graph corresponding to the Cardiff city street network of Figure \ref{fig:cardiff_st} for $t=500$ meters is illustrated in Figure \ref{fig:cardiff_st_r}. In this figure, for the vertex represented by a blue circle, all adjacent vertices in the corresponding reachability graph $G^r_t$, $t=500$, are represented by red circles.

The $k$-domination problem with respect to a street network reachability graph is a useful model for facility location problems \cite{gagarin2018multiple}. By placing the facility in question at the locations corresponding to a $k$-dominating set, we ensure that any agent wishing to use the facility in the street network has a guaranteed minimum level of access options. Furthermore, by minimizing the size of a $k$-dominating set, we minimize the cost of providing this facility.

\begin{figure}
\begin{center}
\subfigure[]{\includegraphics[width=5.5cm]{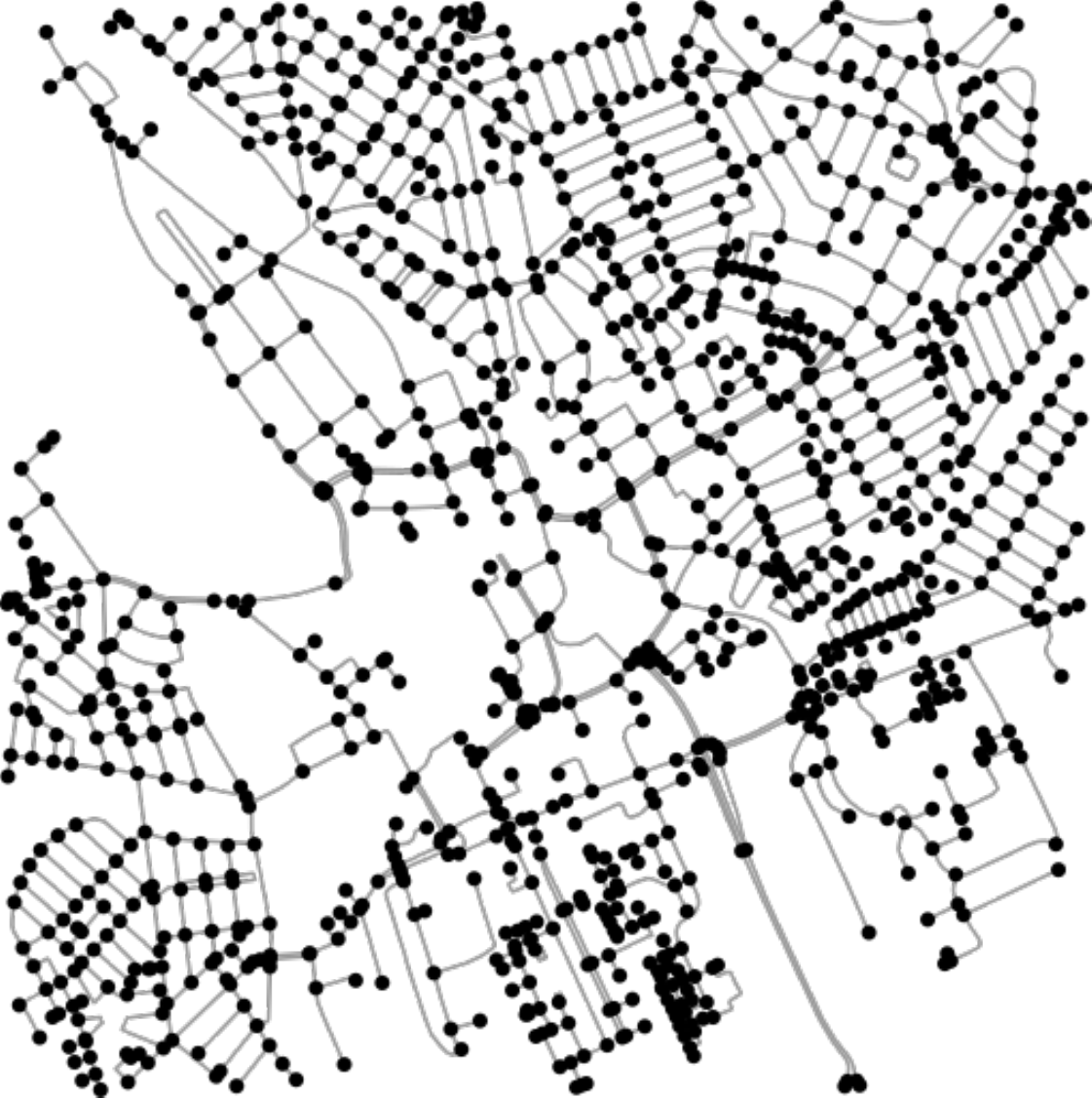}
\label{fig:cardiff_st}}
\hspace{.1cm}
\subfigure[]{\includegraphics[width=5.5cm]{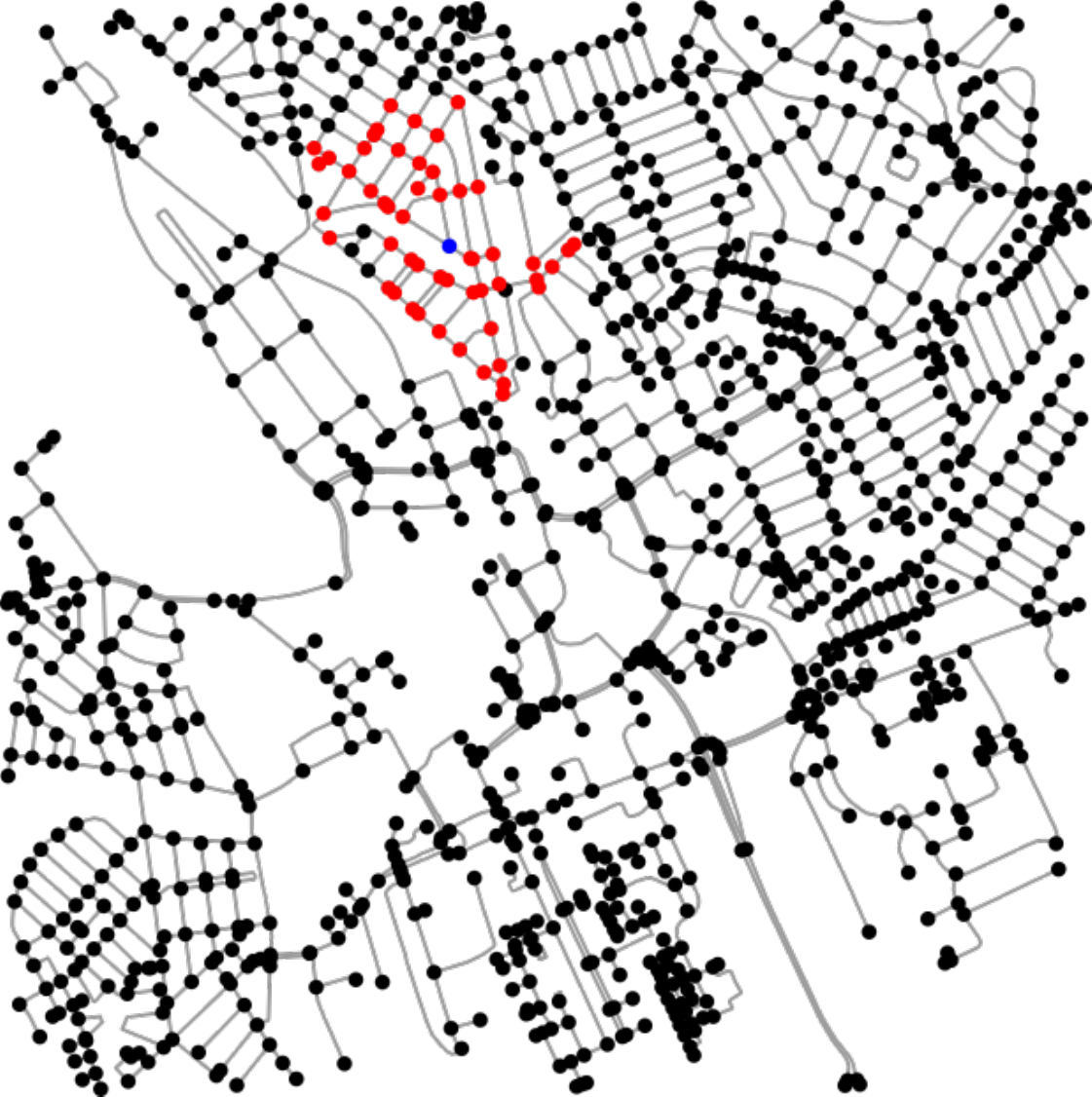}
\label{fig:cardiff_st_r}}
\caption{(a) A graph modelling the street network of Cardiff city; (b) An illustration of reachability graph for the Cardiff city street network: all red vertices are adjacent to a vertex represented by a blue circle.}
\label{fig:cardiff}
\end{center}
\end{figure}

\begin{figure}
\begin{center}
\vspace{5mm}
\includegraphics[width=100mm]{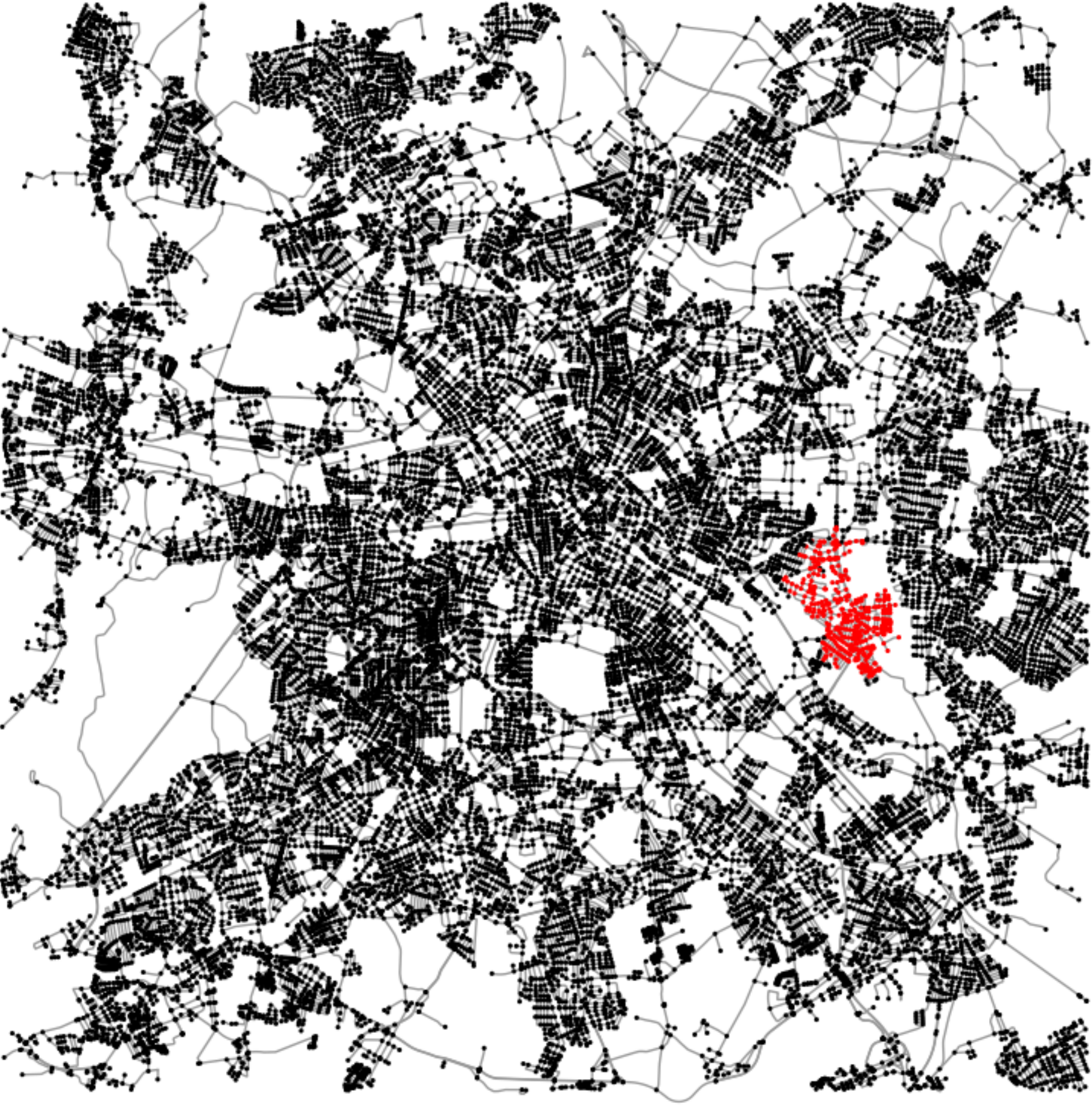}
\caption{A graph modelling the street network of Berlin city. All red vertices are adjacent to a given single vertex in the corresponding reachability graph.}
\label{fig:berlin_st}
\end{center}
\end{figure}

\subsection{Street Networks}
\label{sec:results:stn}
To evaluate the proposed methods with respect to street network reachability graphs we considered 20 medium sized street networks corresponding to twenty UK cities and 5 large sized street networks corresponding to international cities. For each city we selected a location in the city center and extracted the street network graph $G^s$ within a bounding box centred at this location. For each UK and international city a 3 and 15 kilometer bounding box respectively was used. The street networks in question were obtained from OpenStreetMap which is a crowdsourcing project for geographical data \cite{boeing2017osmnx}. For each UK and international street network, the corresponding reachability graph $G^r_t$ was computed using a reachability threshold $t$ of 500 and 3000 meters respectively. The reachability graphs for the cities of Cardiff and Berlin computed using the above approach are illustrated in Figures \ref{fig:cardiff} and \ref{fig:berlin_st} respectively.

Tables \ref{table:city_graph_stats} and \ref{table:city_graph_stats_international} display the names of the UK and international cities respectively, the number of vertices and edges in the corresponding street network graphs $G^s$, and the number of vertices and edges in the corresponding reachability graphs $G^r_t$.

\begin{table}
\centering
\begin{tabular}{ c|c|c|c } 
 \hline
  City Name & No. vertices & No. edges & No. edges \\
  & $G^s$ ($G^r_t$) & $G^s$ & $G^r_t$  \\
 \hline
 Bath & 910 & 1,147 & 18,560 \\
 Belfast & 1,700 & 2,169 & 62,617 \\
 Brighton & 976 & 1,342 & 35,012 \\
 Bristol & 1,569 & 2,048 & 47,522 \\
 Cardiff & 1,127 & 1,466 & 23,155 \\
 Coventry & 1,175 & 1,507 & 26,689 \\
 Exeter & 1,250 & 1,475 & 31,997 \\
 Glasgow & 1,137 & 1,546 &  24,323 \\
 Leeds & 1,647 & 2,197 & 56,511 \\
 Leicester & 1,531 & 2,027 & 48,219 \\
 Liverpool & 1,273 & 1,721 &  42,564 \\
 Manchester & 1,991 & 2,696 & 77,286 \\
 Newcastle & 1,109 & 1,402 & 26,614 \\
 Nottingham & 1,739 & 2,134 & 51,595 \\
 Oxford & 479 & 581 & 8,396 \\
 Plymouth & 1,122 & 1,463 & 35,070 \\
 Sheffield & 1,582 & 2,065 & 50,534 \\
 Southampton & 796 & 1,062 & 19,942 \\
 Sunderland & 1,346 & 1,783 & 42,013 \\
 York & 1,044 & 1,228 & 23,774 \\
 \hline
\end{tabular}
\caption{The number of vertices and edges in the street network graph $G^s$ and the corresponding reachability graph $G^r_t$ for $20$ UK cities.}
\label{table:city_graph_stats}
\end{table}

\begin{table}
\centering
\begin{tabular}{ c|c|c|c } 
 \hline
  City Name & No. vertices & No. edges & No. edges \\
  & $G^s$ ($G^r_t$) & $G^s$ & $G^r_t$  \\
 \hline
 Belgrade, Serbia & 22,218 & 28,465 & 9,092,430 \\
 Berlin, Germany & 31,413 & 46,948 & 10,356,466 \\
 Boston, USA & 34,713 & 50,190 & 23,379,262 \\
 Dublin, Ireland & 35,172 & 41,744 & 20,513,936 \\
 Minsk, Belarus & 11,388 & 16,217 & 1,387,938 \\
 \hline
\end{tabular}
\caption{The number of vertices and edges in the street network graph $G^s$ and the corresponding reachability graph $G^r_t$ for $5$ international cities.}
\label{table:city_graph_stats_international}
\end{table}

\subsection{Baseline Methods}
\label{sec:results:baseline}
We considered the standard heuristic algorithm (``standard greedy") \cite{parekh1991analysis,gagarin2018multiple} and the algorithm by Couture et al. \cite{couture2008incremental} as baseline solution methods. As discussed in the related works section of this paper, these are current state-of-the-art heuristic solution methods for the $k$-domination problem. We now briefly review each of these methods in turn.

The standard heuristic algorithm (``standard greedy") \cite{parekh1991analysis,gagarin2018multiple} initializes a set $D$ to be the empty set and iteratively adds vertices to $D$ until it forms a $k$-dominating set. The vertex added to $D$ at each step is determined by selecting uniformly at random a vertex from the set of vertices whose closed neighbourhood currently contains a maximum number of not dominated enough vertices.

The method of Couture et al. \cite{couture2008incremental} first computes a dominating set ($k=1$) by computing a maximal independent set. Next, it computes a maximal independent set for those vertices that are currently not $2$-dominated, and adds them to the dominating set to form a $2$-dominating set. This procedure is repeated $k$ times until a $k$-dominating set is found. In our implementation, a greedy randomized sequential algorithm was used to compute the maximal independent sets \cite{blelloch2012greedy}.

\subsection{Empirical Results}
\label{sec:results:comparsion}
This section presents an empirical evaluation of the proposed methods for computing $k$-dominating sets in the street network reachability graphs described in Section \ref{sec:results:stn} with respect to the baseline methods described in Section \ref{sec:results:baseline}. 

The beam search heuristic method of Algorithm \ref{alg:rand_beam_dom} has a single hyper-parameter of beam width. Recall that, this method with a beam width of one reduces to the greedy search heuristic method of Algorithm \ref{alg:rand_gred_dom}. For the medium size UK city graphs, we present results with respect to the beam search heuristic method for the three beam widths of $1$, $2$, and $4$. For the large international city graphs, we present results with respect to the greedy search heuristic method. Due to the high computational complexity of the beam search heuristic method, it was not feasible to apply this method to these large graphs. Since all new and baseline methods have a randomized component, they may find dominating sets of different sizes when run with different random seeds. To understand the effect of this randomness, for a given method and graph, we applied the method to the graph using ten different random seeds and reported the minimum, mean and standard deviation statistics of the resulting dominating set sizes. The minimum is a very relevant statistic because when using a randomized algorithm, one typically runs the algorithm multiple times and uses the best result achieved.

Table \ref{table:k_1_results} displays the statistics of the computed dominating set sizes for $k=1$ for each UK city. The last row in this table displays the average of each statistic computed by each method across all cities. From these results, we see that the method of Couture et al. \cite{couture2008incremental} performed less well by a significant margin. Specifically, the average minimum and mean dominating set size is significantly greater than the other methods. Furthermore, the average standard deviation of the dominating set size is also significantly greater than the other methods. This demonstrates that the method is less stable and more dependent on the choice of random seed.

The standard greedy method \cite{parekh1991analysis} performed equally well as the beam search heuristic method with a beam width of one. This result can be attributed to Theorem \ref{thm:mul_equ_methods} which established an equivalence between these methods. For most cities, the smallest minimum and mean dominating set size was achieved when using the beam search heuristic method with a larger beam width. This is reflected in the corresponding average statistics. This demonstrates the usefulness of using a beam search as opposed to the standard greedy search heuristic. Finally, the average standard deviation of the dominating set size for the standard greedy and beam search heuristic methods is quite small. This demonstrates that both methods are quite stable and less dependent on the choice of random seed.

\begin{table}
\small
\centering
\begin{tabular}{ c|c|c|c|c|c } 
 \hline
  City Name & Beam Sr & Beam Sr & Beam Sr & Standard Greedy & Couture et al. \\
  & $b=1$ & $b=2$ & $b=4$ & \cite{parekh1991analysis,gagarin2018multiple} &  \cite{couture2008incremental}  \\
 \hline
 Bath & 44, 45.0, 0.8 & 43, 44.7, 1.0 & 43, 44.6, 0.9 & 44, 45.1, 0.9 & 58, 63.9, 3.1 \\
 Belfast & 48, 50.5, 1.8 & 48, 50.3, 1.6 & 48, 50.2, 1.5 & 48, 50.3, 1.5 & 74, 76.6, 2.7 \\
 Brighton & 28, 28.6, 0.9 & 28, 28.2, 0.6 & 28, 28.2, 0.6 & 28, 28.8, 1.0 & 33, 37.3, 2.2 \\
 Bristol & 47, 47.5, 1.0 & 46, 47.4, 1.0 & 47, 47.2, 0.9 & 47, 47.7, 0.8 & 69, 74.1, 5.6 \\
 Cardiff & 49, 51.0, 0.9 & 49, 50.8, 0.7 & 48, 50.6, 1.0 & 49, 50.8, 0.9 & 71, 76.8, 4.0 \\
 Coventry & 44, 45.1, 0.5 & 44, 44.9, 0.3 & 44, 44.8, 0.4 & 44, 45.2, 0.6 & 69, 73.8, 3.1 \\
 Exeter & 50, 50.8, 0.9 & 49, 50.6, 0.8 & 50, 50.6, 0.5 & 50, 51.0, 0.7 & 71, 76.7, 3.0 \\
 Glasgow & 58, 59.7, 1.2 & 58, 59.5, 1.1 & 58, 59.2, 0.7 & 58, 59.7, 1.2 & 79, 86.4, 3.7 \\
 Leeds & 51, 52.7, 0.8 & 51, 52.6, 0.8 & 51, 52.4, 0.8 & 51, 52.7, 0.8 & 73, 76.7, 3.6 \\
 Leicester & 51, 51.8, 0.4 & 51, 51.6, 0.5 & 51, 51.5, 0.5 & 51, 52.0, 0.4 & 75, 80.9, 3.6 \\
 Liverpool & 38, 38.4, 0.5 & 38, 38.5, 0.5 & 38, 38.4, 0.5 & 38, 38.6, 0.5 & 50, 56.6, 4.3 \\
 Manchester & 45, 46.2, 0.9 & 45, 46.0, 0.6 & 45, 45.9, 0.5 & 45, 46.0, 0.8 & 71, 75.2, 3.7 \\
 Newcastle & 52, 53.3, 1.0 & 51, 52.9, 1.1 & 51, 52.6, 1.1 & 52, 53.4, 0.6 & 73, 77.5, 2.4 \\
 Nottingham & 56, 57.1, 0.8 & 56, 56.9, 0.7 & 55, 56.6, 0.8 & 56, 56.9, 0.5 & 77, 81.8, 2.6 \\
 Oxford & 28, 28.2, 0.4 & 27, 28.0, 0.6 & 27, 27.9, 0.5 & 27, 28.1, 0.7 & 38, 40.9, 1.9 \\
 Plymouth & 40, 40.6, 0.6 & 40, 40.5, 0.7 & 39, 40.3, 0.8 & 40, 40.7, 0.6 & 54, 59.1, 3.4 \\
 Sheffield & 52, 53.3, 0.9 & 52, 53.0, 0.9 & 51, 52.5, 0.7 & 52, 53.1, 0.5 & 76, 81.4, 3.2 \\
 Southampton & 29, 29.9, 0.8 & 28, 29.8, 0.9 & 28, 29.6, 0.8 & 29, 29.6, 0.5 & 41, 46.0, 2.7 \\
 Sunderland & 46, 46.5, 0.5 & 46, 46.3, 0.4 & 46, 46.3, 0.4 & 46, 46.5, 0.5 & 56, 62.7, 3.7 \\
 York & 39, 39.3, 0.4 & 39, 39.2, 0.4 & 39, 39.1, 0.3 & 39, 39.5, 0.5 & 60, 68.5, 4.7 \\
 \hline
 \rule{0pt}{4ex} Average & 44.7, 45.7, 0.8 & 44.4, 45.5, 0.7 & 44.3, 45.4, 0.7 & 44.7, 45.7, 0.7 & 63.4, 68.6, 3.3 \\
 \hline
\end{tabular}
\caption{The minimum, mean and standard deviation of the dominating set ($k=1$) sizes computed using different heuristic methods for 20 UK cities.}
\label{table:k_1_results}
\end{table}

Tables \ref{table:k_2_results} and \ref{table:k_4_results} display the statistics of the computed dominating set sizes for $k$ equal to $2$ and $4$ respectively for each UK city. For both values of $k$, the method of Couture et al. \cite{couture2008incremental} performed less well by a significant margin while the beam search heuristic method performed the best. Comparing the averages of the best found $2$-dominating set sizes, we see that on average it achieved a best $2$-dominating set size approximately $3$, $3.5$, and $4$ vertices smaller than that achieved by the standard greedy method when using the beam width of 1, 2, and 4, respectively. This approximately equals a 5\% reduction in the size of best found $2$-dominating sets. Comparing the averages of the best found $4$-dominating set sizes, we see that on average it achieved a best $4$-dominating set size approximately $13.5$, $14$, and $14.5$ vertices smaller than that achieved by the standard greedy method when using the beam width of 1, 2, and 4, respectively. This approximately equals a 9\% reduction in the size of best found $2$-dominating sets. Similarly to the case $k=1$, the average standard deviation of the dominating set size for the standard greedy and beam search heuristic methods is quite small.

\begin{table}
\small
\centering
\begin{tabular}{ c|c|c|c|c|c } 
 \hline
  City Name & Beam Sr & Beam Sr & Beam Sr & Standard Greedy & Couture et al. \\
  & $b=1$ & $b=2$ & $b=4$ & \cite{gagarin2018multiple} &  \cite{couture2008incremental}  \\
 \hline
 Bath & 86, 89.0, 1.4 & 87, 89.0, 0.6 & 87, 88.0, 0.6 & 90, 91.2, 0.7 & 118, 122.5, 3.1 \\
 Belfast & 96, 98.9, 1.6 & 96, 98.8, 1.6 & 96, 97.6, 1.0 & 100, 104.5, 1.9 & 140, 147.9, 4.0 \\
 Brighton & 50, 50.6, 0.5 & 49, 50.0, 0.6 & 49, 49.4, 0.5 & 51, 52.3, 1.0 & 68, 73.2, 3.8 \\
 Bristol & 93, 95.2, 1.1 & 93, 94.8, 0.9 & 91, 94.0, 1.4 & 96, 98.2, 1.8 & 142, 145.1, 4.0 \\
 Cardiff & 95, 97.7, 1.4 & 95, 97.1, 1.1 & 92, 95.9, 1.6 & 97, 98.8, 0.9 & 141, 146.1, 3.7 \\
 Coventry & 85, 85.8, 0.7 & 84, 85.3, 0.6 & 84, 85.1, 0.7 & 88, 89.1, 1.3 & 138, 142.3, 3.8 \\
 Exeter & 94, 96.4, 1.1 & 94, 96.1, 0.9 & 94, 95.7, 1.0 & 97, 99.3, 1.6 & 141, 148.4, 4.5 \\
 Glasgow & 111, 112.5, 1.3 & 108, 111.6, 1.7 & 108, 110.6, 1.7 & 113, 116.3, 2.2 & 149, 157.0, 4.4 \\
 Leeds & 99, 100.3, 0.6 & 99, 100.0, 0.6 & 98, 99.6, 1.0 & 101, 102.6, 2.3 & 143, 150.2, 5.7 \\
 Leicester & 94, 94.8, 0.6 & 93, 94.4, 0.9 & 93, 94.1, 0.8 & 100, 101.0, 0.8 & 146, 150.3, 2.9 \\
 Liverpool & 71, 72.4, 0.8 & 71, 72.4, 0.8 & 71, 72.0, 0.8 & 73, 74.4, 0.6 & 102, 110.0, 4.6 \\
 Manchester & 92, 93.0, 0.8 & 90, 92.2, 0.9 & 90, 91.5, 0.9 & 92, 93.9, 1.9 & 143, 147.5, 3.3 \\
 Newcastle & 95, 97.2, 1.8 & 95, 96.4, 1.4 & 94, 95.4, 1.1 & 99, 101.5, 1.2 & 133, 142.3, 4.8 \\
 Nottingham & 102, 103.5, 0.8 & 102, 103.3, 0.8 & 102, 103.3, 0.8 & 107, 108.5, 0.8 & 156, 161.0, 4.3 \\
 Oxford & 55, 55.5, 0.7 & 54, 55.2, 0.6 & 54, 54.9, 0.7 & 58, 59.4, 1.0 & 74, 76.8, 2.2 \\
 Plymouth & 74, 76.3, 1.2 & 74, 75.8, 1.0 & 73, 75.0, 1.1 & 77, 78.5, 0.7 & 105, 111.2, 4.9 \\
 Sheffield & 98, 100, 1.4 & 97, 99.5, 1.5 & 97, 98.9, 1.3 & 105, 106.7, 1.7 & 148, 155.1, 5.4 \\
 Southampton & 61, 61.7, 0.6 & 61, 61.6, 0.5 & 60, 61.1, 0.7 & 62, 64.2, 1.0 & 78, 87.2, 5.5 \\
 Sunderland & 88, 90.4, 1.4 & 88, 89.9, 1.1 & 87, 89.1, 1.1 & 91, 92.2, 0.6 & 120, 121.8, 2.6 \\
 York & 77, 77.9, 0.7 & 77, 78.0, 0.8 & 77, 77.6, 0.6 & 78, 78.8, 0.7 & 127, 131.0, 2.9 \\
 \hline 
 \rule{0pt}{4ex} Average & 85.8, 87.4, 1.0 & 85.3, 87.0, 0.9 & 84.8, 86.4, 0.9 & 88.7, 90.5, 1.2 & 125.6, 131.3, 4.0 \\
 \hline
\end{tabular}
\caption{The minimum, mean and standard deviation of the dominating set ($k=2$) sizes computed using different heuristic methods for 20 UK cities.}
\label{table:k_2_results}
\end{table}

\begin{table}
\small
\centering
\begin{tabular}{ c|c|c|c|c|c } 
 \hline
  City Name & Beam Sr & Beam Sr & Beam Sr & Standard Greedy & Couture et al. \\
  & $b=1$ & $b=2$ & $b=4$ & \cite{gagarin2018multiple} &  \cite{couture2008incremental}  \\
 \hline
 Bath & 160, 162.8, 1.6 & 159, 161.3, 1.2 & 159, 160, 1.1 & 178, 180, 1.33 & 210, 220.4, 4.3 \\
 Belfast & 178, 181.0, 1.7 & 177, 180.2, 2.0 & 177, 179.6, 2.0 & 194, 196.0, 1.4 & 257, 263.0, 4.4 \\
 Brighton & 93, 95.7, 1.4 & 93, 94.4, 0.9 & 92, 94.8, 1.9 & 101, 103.5, 2.1 & 123, 135.4, 5.5 \\
 Bristol & 176, 177.7, 1.1 & 175, 176.8, 0.9 & 175, 176.4, 0.8 & 187, 188.3, 0.9 & 253, 263.2, 5.4 \\
 Cardiff & 181, 185.0, 1.7 & 181, 183.6, 1.7 & 181, 183.2, 1.4 & 196, 199.6, 2.0 & 238, 252.5, 7.3 \\
 Coventry & 171, 175.1, 1.7 & 171, 174.1, 1.5 & 170, 172.6, 1.4 & 182, 183.4, 1.5 & 247, 255.2, 5.2 \\
 Exeter & 182, 183.1, 0.9 & 182, 182.8, 0.6 & 181, 182.3, 0.6 & 196, 199.4, 1.9 & 259, 263.9, 3.2 \\
 Glasgow & 198, 201.6, 1.9 & 198, 200.5, 1.6 & 197, 199.8, 1.6 & 221, 226.2, 2.3 & 256, 264.3, 4.3 \\
 Leeds & 186, 188.5, 1.2 & 187, 188.0, 0.8 & 186, 187.1, 0.7 & 198, 201.5, 2.4 & 264, 268.6, 3.8 \\
 Leicester & 176, 179.6, 1.4 & 176, 179.3, 1.6 & 175, 177.7, 1.8 & 199, 202.0, 2.3 & 267, 274.1, 4.5 \\
 Liverpool & 133, 134.5, 1.4 & 132, 133.7, 1.1 & 132, 133, 0.8 & 143, 145.4, 1.5 & 194, 201.9, 4.9 \\
 Manchester & 177, 179.6, 1.2 & 177, 179.1, 1.2 & 177, 178.5, 1.0 & 185, 188.5, 1.3 & 251, 266.6, 7.6 \\
 Newcastle & 170, 172.4, 1.0 & 169, 171.5, 1.2 & 170, 171.2, 0.7 & 189, 192.9, 2.33 & 242, 246.6, 4.9 \\
 Nottingham & 194, 196.5, 1.1 & 194, 195.3, 1.0 & 193, 195.2, 1.2 & 205, 208.4, 2.4 & 288, 295.1, 4.6 \\
 Oxford & 100, 101.7, 1.2 & 99, 100.8, 0.9 & 100, 100.8, 0.8 & 108, 114.8, 3.1 & 129, 130.9, 1.2 \\
 Plymouth & 137, 138.8, 1.3 & 136, 137.9, 1.1 & 135, 137.0, 1.2 & 153, 155.1, 1.4 & 195, 200.4, 4.0 \\
 Sheffield & 182, 184.3, 0.9 & 182, 183.2, 1.1 & 180, 182.2, 1.2 & 202, 204.3, 1.5 & 272, 278.3, 4.3 \\
 Southampton & 113, 114.7, 1.6 & 113, 114.2, 1.3 & 112, 113.2, 1.4 & 124, 125.2, 1.2 & 150, 156.0, 3.4 \\
 Sunderland & 164, 164.8, 0.7 & 163, 164.1, 0.7 & 162, 163.6, 1.0 & 176, 180.7, 3.1 & 218, 223.8, 4.2 \\
 York & 146, 147.1, 1.0 & 145, 146.4, 1.2 & 144, 145.8, 1.2 & 153, 157.2, 2.2 & 223, 231.0, 6.0 \\
 \hline 
 \rule{0pt}{4ex} Average & 160.8, 163.2, 1.3 & 160.4, 162.3, 1.1 & 159.9, 161.7, 1.1 & 174.5, 177.6, 1.9 & 226.8, 234.5, 4.6 \\
 \hline
\end{tabular}
\caption{The minimum, mean and standard deviation of the dominating set ($k=4$) sizes computed using different heuristic methods for 20 UK cities.}
\label{table:k_4_results}
\end{table}

Tables \ref{table:k_1_results_international}, \ref{table:k_2_results_international}, and \ref{table:k_4_results_international} display the statistics of the computed dominating set sizes for $k$ equal to $1$, $2$, and $4$ respecively for each international city. For all values of $k$, Couture et al. \cite{couture2008incremental} performed less well by a significant margin. For $k$ equal to $2$ and $4$, the proposed greedy heuristic method performed the best. In fact, comparing the averages of the best found dominating set sizes, we see that on average it achieved a best $2$- and $4$-dominating set size approximately $6$ and $23$ vertices respectively smaller than that achieved by the standard greedy method. This approximately equals a 3\% and 6\% reduction respectively in the size of best found $k$-dominating set.

\begin{table}
\centering
\begin{tabular}{ c|c|c|c } 
 \hline
  City Name & Proposed Greedy & Standard Greedy & Couture et al. \\
  &  & \cite{parekh1991analysis,gagarin2018multiple} &  \cite{couture2008incremental}  \\
 \hline
 Belgrade, Serbia & 99, 100.3, 0.9 & 99, 100.9, 0.9 & 128, 138.8, 5.0 \\
 Berlin, Germany & 146, 147.2, 1.1 & 146, 147.0, 0.7 & 174, 181.7, 6.0  \\
 Boston, USA & 71, 72.8, 1.2 & 71, 72.6, 1.4 & 85, 86.8, 2.3 \\
 Dublin, Ireland & 88, 90.2, 1.8 & 89, 89.8, 1.2 & 129, 134.1, 3.5 \\
 Minsk, Belarus & 139, 139.8, 0.9 & 138, 139.3, 1.1 & 183, 188.9, 4.6 \\
 \hline
 \rule{0pt}{4ex} Average & 108.6, 109.9, 1.1 & 108.6, 109.9, 1.0 & 139.8, 146.0, 4.2 \\
 \hline
\end{tabular}
\caption{The minimum, mean and standard deviation of the dominating set ($k=1$) sizes computed using different heuristic methods for 5 international cities.}
\label{table:k_1_results_international}
\end{table}

\begin{table}
\centering
\begin{tabular}{ c|c|c|c } 
 \hline
  City Name & Proposed Greedy & Standard Greedy & Couture et al. \\
  &  & \cite{parekh1991analysis,gagarin2018multiple} &  \cite{couture2008incremental}  \\
 \hline
 Belgrade, Serbia & 197, 197.5, 0.9 & 199, 200.1, 0.7 & 271, 281.2, 6.7 \\
 Berlin, Germany & 268, 270.4, 1.4 & 277, 278.0, 1.6 & 352, 361.5, 7.4 \\
 Boston, USA & 133, 135.0, 1.1 & 137, 138.0, 0.6 & 167, 175.7, 6.1 \\
 Dublin, Ireland & 166, 166.8, 0.4 & 175, 175.8, 0.7 & 256, 269.0, 7.6 \\
 Minsk, Belarus & 265, 266.4, 1.3 & 271, 272.7, 1.3 & 360, 365.4, 5.4 \\
 \hline
 \rule{0pt}{4ex} Average & 205.8, 207.2, 1.0 & 211.8, 212.9, 0.9 & 281.2, 290.5, 6.6 \\
 \hline
\end{tabular}
\caption{The minimum, mean and standard deviation of the dominating set ($k=2$) sizes computed using different heuristic methods for 5 international cities.}
\label{table:k_2_results_international}
\end{table}

\begin{table}
\centering
\begin{tabular}{ c|c|c|c } 
 \hline
  City Name & Proposed Greedy & Standard Greedy & Couture et al. \\
  &  & \cite{parekh1991analysis,gagarin2018multiple} &  \cite{couture2008incremental}  \\
 \hline
 Belgrade, Serbia & 379, 381.4, 2.0 & 397, 400.1, 1.3 & 537, 550.5, 8.5 \\
 Berlin, Germany & 501, 503.4, 1.9 & 531, 537.4, 4.4 & 685, 697.7, 13.4 \\
 Boston, USA & 255, 255.4, 0.5 & 265, 266.8, 2.0 & 336, 349.4, 7.7 \\
 Dublin, Ireland & 317, 318.5, 1.5 & 342, 344.1, 1.6 & 514, 530.5, 10.4 \\
 Minsk, Belarus & 506, 509.4, 2.9 & 539, 542.3, 2.4 & 684, 701.5, 9.6 \\
 \hline
 \rule{0pt}{4ex} Average & 391.6, 393.62, 1.76 & 414.8, 418.14, 2.34 & 551.2, 565.92, 9.92 \\
 \hline
\end{tabular}
\caption{The minimum, mean and standard deviation of the dominating set ($k=4$) sizes computed using different heuristic methods for 5 international cities.}
\label{table:k_4_results_international}
\end{table}

The proposed greedy and beam search heuristic methods perform better than the standard greedy approach for larger values of $k$. This can be attributed to the fact that, as the value of $k$ increases, the number of levels to which a vertex can be dominated increases. For example, in the case $k=1$, a vertex can only be dominated or not dominated. On the other hand, when $k=4$, a vertex can be dominated to five different levels, corresponding to the values of coverage parameter $C(D,v)$ in problem formulation (\ref{eq:dom_opt_re}). The proposed greedy and beam search heuristic methods exploit this information to make better decisions, while the standard greedy approach does not. In summary, for the $k$-domination problem with $k>1$, the proposed new greedy and beam search heuristic methods outperform the baseline methods, and the performance gain is greater for larger values of $k$.

Table \ref{table:k_2_times} reports running times measured in seconds required by the proposed and baseline methods to compute $2$-dominating sets for five UK cities and five international cities. All algorithms were implemented in the Python programming language and executed on a desktop computer with an Intel Core i7-8700 CPU. The proposed and standard greedy algorithms run very quickly on the medium sized UK city networks. Both algorithms run reasonably fast on the large sized international city networks considering the size of the networks in question. The general beam search heuristic method runs much slower than either of the greedy algorithms. This can be attributed to higher computational complexity plus the challenge in transforming this method into an efficient implementation. 
In particular, although the beam search heuristic with the beam width equal to $1$ is equivalent to the proposed greedy algorithm, its overhead makes it much slower and inefficient in comparison to the pure greedy version. 
The method of \cite{couture2008incremental} runs very quickly on the medium sized UK city networks as well as the large sized international city networks.

\begin{table}
\small
\centering
\begin{tabular}{ c|c|c|c|c|c|c } 
 \hline
  City Name & Proposed & Beam Sr & Beam Sr & Beam Sr & Standard & Couture \\
   & Greedy & $b=1$ & $b=2$ & $b=4$ & Greedy \cite{gagarin2018multiple} & et al. \cite{couture2008incremental} \\
 \hline
 Bath & 1, 0 & 112, 1 & 464, 15 & 1736, 41 & 1, 0 & 1, 0 \\
 Belfast & 3, 0 & 582, 14 & 2182, 99 & 8834, 177 & 2, 0 & 1, 0 \\
 Brighton & 1, 0 & 78, 2 & 309, 4 & 1257, 31 & 1, 0 & 1, 0 \\
 Bristol & 2, 0 & 447, 8 & 1752, 30 & 7156, 161 & 2, 0 & 1, 0 \\
 Cardiff & 1, 0 & 210, 2 & 835, 21 & 3327, 100 & 1, 0 & 1, 0 \\
 Belgrade & 907, 11 & -  & - & - & 800, 17  & 14, 1 \\
 Berlin & 1435, 40 & - & - & - & 1199, 31  & 21, 1 \\
 Boston & 1761, 21 & - & - & - & 1437, 21 & 23, 2  \\
 Dublin & 1828, 9 & - & - & - & 1593, 18  & 23, 1  \\
 Minsk & 188, 1 & - & - & - & 165, 3  & 7, 0  \\
 \hline
\end{tabular}
\caption{The mean and standard deviation of running times (in seconds) for computing dominating sets ($k=2$) using different heuristic methods.}
\label{table:k_2_times}
\end{table}

\section{Conclusion}
\label{sec:conclusions}
In this work, we proposed novel greedy and beam search heuristic methods for the $k$-domination problem. These methods are inspired by a novel formulation of the problem (\ref{eq:dom_opt_re}). The methods were evaluated with respect to two baseline methods on a set of street network reachability graphs. Our evaluation found that, for the classic domination problem ($k=1$), the proposed methods perform equally well with one of the existing methods. This result is attributed to an equivalence between methods in this particular case. On the other hand, for the $k$-domination problem with $k>1$, the proposed methods outperform the baseline methods, and the performance gain is greater for larger values of $k$.

A useful characteristic of the proposed methods is their simplicity. The proposed beam search heuristic method of Algorithm \ref{alg:rand_beam_dom} with a beam width of one reduces to the greedy search heuristic method of Algorithm \ref{alg:rand_gred_dom}. The latter algorithm is efficient and simple to implement. This contrasts with metaheuristic or machine learning based methods for combinatorial optimization problems which can be very challenging to implement \cite{khalil2017learning}. 

Possible directions for future research in this area include the following. The evaluation presented in this work was purely empirical. In future work, it would be interesting to show better analysis of the performance of the proposed methods. Such analysis could take the form of proving some bounds on the size of the $k$-dominating sets found by the algorithms. The related works section of this article highlighted that there currently exist no metaheuristic methods for the $k$-domination problem where $k>1$. Given good performance of such methods with respect to the classic domination problem ($k=1$), this presents an interesting direction for research as well.

\bibliographystyle{unsrt}
\bibliography{multiple_domination_heuristic}

\end{document}